\def\useieeelayout{1}
\def\showall{0}

\newcommand{\inConf}[1]{\if\showall1{\color{green!50!black}In Journal: #1}\else\if\useieeelayout1{#1}\fi\fi}

\newcommand{\inArxiv}[1]{\if\showall1{\color{blue}In ArXiV: #1}\else\if\useieeelayout0{#1}\fi\fi}

\if\useieeelayout1
\documentclass[letterpaper, 10 pt, conference]{ieeeconf}  
\IEEEoverridecommandlockouts                              %
\usepackage{graphicx}
\usepackage{amsmath}
\usepackage{optidef}
 
\usepackage{amsthm}
\usepackage{amssymb}
\usepackage{amsfonts}       
\usepackage{multicol}
\usepackage{cuted}
\usepackage{multirow}
\usepackage{caption, subcaption}
\newtheorem{theorem}{Theorem}
\newtheorem{lemma}{Lemma}
\newtheorem{proposition}{Proposition}
\newtheorem{definition}{Definition}
\newtheorem{problem}{Problem}
\newtheorem{assumption}{Assumption}
\newtheorem{remark}{Remark}
\newtheorem{conjecture}[theorem]{Conjecture}
\newtheorem{question}{Question}
\newtheorem{corollary}{Corollary}
\usepackage{xcolor}
\usepackage{cite}
\usepackage[colorlinks=true,allcolors=steelblue]{hyperref}
\usepackage{comment}
\usepackage{algorithm}
\usepackage{algorithmic}
\usepackage{scalerel}
\usepackage{bm}
\usepackage{tikz}
\usepackage{tkz-euclide}
\usetikzlibrary{calc,fadings,decorations.pathreplacing,arrows,positioning}
\usepackage{3dplot}
\usepackage{pgfplots}
\definecolor{darkgreen}{rgb}{0.125,0.5,0.169}
\tikzset{%
  >=latex,
  inner sep=0pt,%
  outer sep=2pt,%
  mark coordinate/.style={inner sep=0pt,outer sep=0pt,minimum size=3pt,
    fill=black,circle}%
}

\definecolor{steelblue}{RGB}{70,130,180}

\def\bbn{\mathbb N}
\def\bbz{\mathbb Z}
\def\bbr{\mathbb R}

\newcommand{\absval}[1]{\mid {#1} \mid}

\newcommand{\norm}[1]{\left\lVert {#1} \right\rVert}

\renewcommand{\b}[1]{{\color{blue}{#1}}}

\title{\LARGE \bf
Stereographic Projection of Probabilistic \\ Frequency-Domain Uncertainty
}

\author{Anton Nystr\"om, Venkatraman Renganathan, and Michael Cantoni
\thanks{A. Nystr\"om and V. Renganathan contributed equally towards this article. This research project has received funding from the European Research Council (ERC) under the European Union’s Horizon 2020 research and innovation program under grant agreement No 834142 (Scalable Control). A. Nystr\"om is a M.Sc. student in the Faculty of Engineering in Lund University, Sweden. V. Renganathan is with the Cranfield University, UK and this work was done when he was with the Department of Automatic Control - LTH, Lund University, Sweden while also being a member of the ELLIIT Strategic Research Area in Lund University.  M. Cantoni is with the Department of Electrical and Electronic Engineering, The University of Melbourne, Australia. Emails: anton.e.nystrom@gmail.com, v.renganathan@cranfield.ac.uk, cantoni@unimelb.edu.au.}%
}

\else
\documentclass[10pt]{article}
\usepackage{graphicx}
\usepackage{amsmath}
\usepackage{optidef}
 
\usepackage{amsthm}
\usepackage{amssymb}
\usepackage{amsfonts}       
\usepackage{multicol}
\usepackage{cuted}
\usepackage{multirow}
\usepackage{caption, subcaption}
\newtheorem{theorem}{Theorem}
\newtheorem{lemma}{Lemma}
\newtheorem{proposition}{Proposition}
\newtheorem{definition}{Definition}
\newtheorem{problem}{Problem}

\newtheorem{question}{Question}

\usepackage{xcolor}
\usepackage{cite}
\usepackage[colorlinks=true,allcolors=steelblue]{hyperref}
\usepackage{comment}
\usepackage{algorithm}
\usepackage{algorithmic}
\usepackage{scalerel}
\usepackage{bm}
\usepackage{tikz}
\usepackage{tkz-euclide}
\usetikzlibrary{calc,fadings,decorations.pathreplacing,arrows,positioning}
\usepackage{3dplot}
\usepackage{pgfplots}
\definecolor{darkgreen}{rgb}{0.125,0.5,0.169}
\tikzset{%
  >=latex,
  inner sep=0pt,%
  outer sep=2pt,%
  mark coordinate/.style={inner sep=0pt,outer sep=0pt,minimum size=3pt,
    fill=black,circle}%
}

\usepackage{url}
\definecolor{steelblue}{RGB}{70,130,180}

\usepackage[preprint]{tmlr}

\title{Stereographic Projection of Probabilistic Frequency-Domain Uncertainty}

\author{\name Anton Nystr\"om 
\email anton.e.nystrom@gmail.com \\
\addr Faculty of Engineering \\ Lund University, Sweden.
\AND
\name Venkatraman Renganathan 
\email v.renganathan@cranfield.ac.uk \\
\addr Cranfield University, UK.
\AND
\name Michael Cantoni 
\email cantoni@unimelb.edu.au \\
\addr Department of Electrical and Electronic Engineering, \\ The University of Melbourne, Australia.
}

\def\month{MM}  
\def\year{YYYY} 

\fi

\begin{document}

\maketitle
\thispagestyle{empty}
\pagestyle{empty}

\begin{abstract}

This paper investigates the stereographic projection of points along the Nyquist plots of single input single output (SISO) linear time invariant (LTI) systems subject to probabilistic uncertainty. At each frequency, there corresponds a complex-valued random variable with given probability distribution in the complex plane. 
The chordal distance between the stereographic projections of this complex value and the corresponding value for a nominal model, as per the well-known $\nu$-Gap metric of Vinnicombe, is also a random quantity. The main result provides the cumulative density function (CDF) of the chordal distance at a given frequency.
Such a stochastic distance framework opens up a fresh and a fertile research direction on probabilistic robust control theory.  
\end{abstract}

\section{Research Motivation}\label{sec_motivation}
A control system is said to be \emph{robust} if it is insensitive to differences between the actual system and the model of the system which is used to design the controller. 
To render a system robustly stable, 
various control strategies based on the small gain theorem\cite{desoer2009feedback} and $\mathcal{H}_{\infty}$ control theory \cite{bacsar2008h}, and integral quadratic constraints \cite{Anders_IQC}, 
have been developed over the years. Conservative modelling of uncertainties can adversely affect the performance of a control system. The above mentioned deterministic robust control approaches give performance guarantees with certainty for the worst case assumptions on the  
modeling uncertainty. 

From the control design perspective, we want the robust feedback compensator of the nominal plant to work well for the perturbed plant too. For establishing robust stability results, researchers came up with several distance metrics between LTI dynamical systems such as the Gap metric \cite{zames1980unstable, georgiou1988computation, georgiou1989optimal}, and Graph metric \cite{vidyasagar1984graph}. Building on top of their works, Vinnicombe proposed the $\nu$-Gap metric in  \cite{vinnicombePhDThesis, vinnicombe2000uncertainty} having sharper quantitative results than the Gap metric. Specifically for the SISO case, the $\nu$-Gap metric
offers a striking benefit of a direct frequency domain interpretation in terms of the chordal distance between the stereographic projection of the Nyquist plots onto the Riemann sphere. Further, while analysing the robustness of feedback systems in \cite{vinnicombePhDThesis}, Vinnicombe presents an interesting problem namely, \emph{how much do we need to know about a system in order to design a feedback compensator that leaves the closed loop insensitive to what we don't know?} The following robust stability result that answers this question is restated from \cite{vinnicombe_tac_1993} in terms of the $\nu$-Gap metric:
\begin{proposition}
\label{proposition_1}
(From \cite{vinnicombe_tac_1993})
Given a nominal continuous time LTI plant $\bar{P}$, and nominal feedback compensator $\bar{C}$, let 
\begin{subequations}
\begin{align}
\label{eqn_bpc_measure}
&b_{\bar{P}, \bar{C}} 
:= 
\begin{cases}
\norm{H(\bar{P},\bar{C})}^{-1}_{\infty}, &\text{if $H(\bar{P},\bar{C})$ is stable} \\
0, &\text{otherwise},
\end{cases}
\intertext{where} 
&H(\bar{P},\bar{C})
:=
\begin{bmatrix}
\bar{P} \\ I    
\end{bmatrix}
(I-\bar{C}\bar{P})^{-1}
\begin{bmatrix}
-\bar{C} & I    
\end{bmatrix},
\label{eqn_gang_of_four}
\end{align}
\end{subequations}
and $\|\cdot\|_\infty$ denotes the $\mathcal{H}_\infty$ norm.
Then, any controller $\bar{C}$ that 
achieves $b_{\bar{P}, \bar{C}} > \alpha$ stabilises the set of plants $\{ P : \delta_{\nu}(P, \bar{P}) \leq \alpha\}$ and
\begin{align} 
    b_{P,\bar{C}} 
    \geq 
    b_{\bar{P},\bar{C}} - \delta_{\nu}(P, \bar{P}), \label{eqn_bpc_performance_degrade}
\end{align}
where $\delta_\nu(P,\bar{P})$ denotes the $\nu$-gap between $P$ and $\bar{P}$.
\end{proposition}
A point-wise-in-frequency version of \eqref{eqn_bpc_performance_degrade} 
also holds as described in
\cite{vinnicombe2000uncertainty,date2003lower}:
\begin{align}
\label{eqn_bpc_performance_degrade_pointwise}
\rho(P(j \omega),\bar{C}(j \omega))
\geq 
\rho(\bar{P}(j \omega),\bar{C}(j \omega)) - \kappa(P(j \omega), \bar{P}(j \omega)),    
\end{align} 
where $\rho(\bar{P}(j \omega),\bar{C}(j \omega)) := 1/\bar{\sigma}(H(\bar{P}(j\omega), \bar{C}(j\omega)))$, and 
\begin{align}
\label{eqn_kappa_siso}
    \kappa(P(j \omega), \bar{P}(j \omega)) = \frac{|P(j\omega)-\bar{P}(j\omega)|}{\sqrt{1+|P(j\omega)|^2}\sqrt{1+|\bar{P}(j\omega)|^2}},
\end{align} 
which 
is 
the chordal distance between stereographic projections of $P(j\omega)$ and $\bar{P}(j\omega)$ onto the Riemann sphere as elaborated in \S\ref{sec_main_SISO}.
From Proposition~\ref{proposition_1},
when $\bar{P}$ and $P$ are \emph{close} in the $\nu$-gap metric, then any stabilizing compensator for the former achieves similar closed loop performance in terms of the measure of performance $b_{P, C}$.
 %

Suppose that $P$ is subject to probabilistic uncertainty with respect to a nominal model $\bar{P}$, which may arise within a system identification context, for example, as motivated subsequently. In line with the probabilistic robust control theory of \cite{NASA_report, Fabrizio_ACC}, one could ask the following question: 
\begin{question} \label{question_1}
Given an identified nominal model, as one realization of an uncertain system, with what probability will a controller that stabilizes the nominal model also stabilize another randomly drawn realization of the uncertain system? 
\end{question}

\begin{figure}
    \centering
    \includegraphics[width=\linewidth]{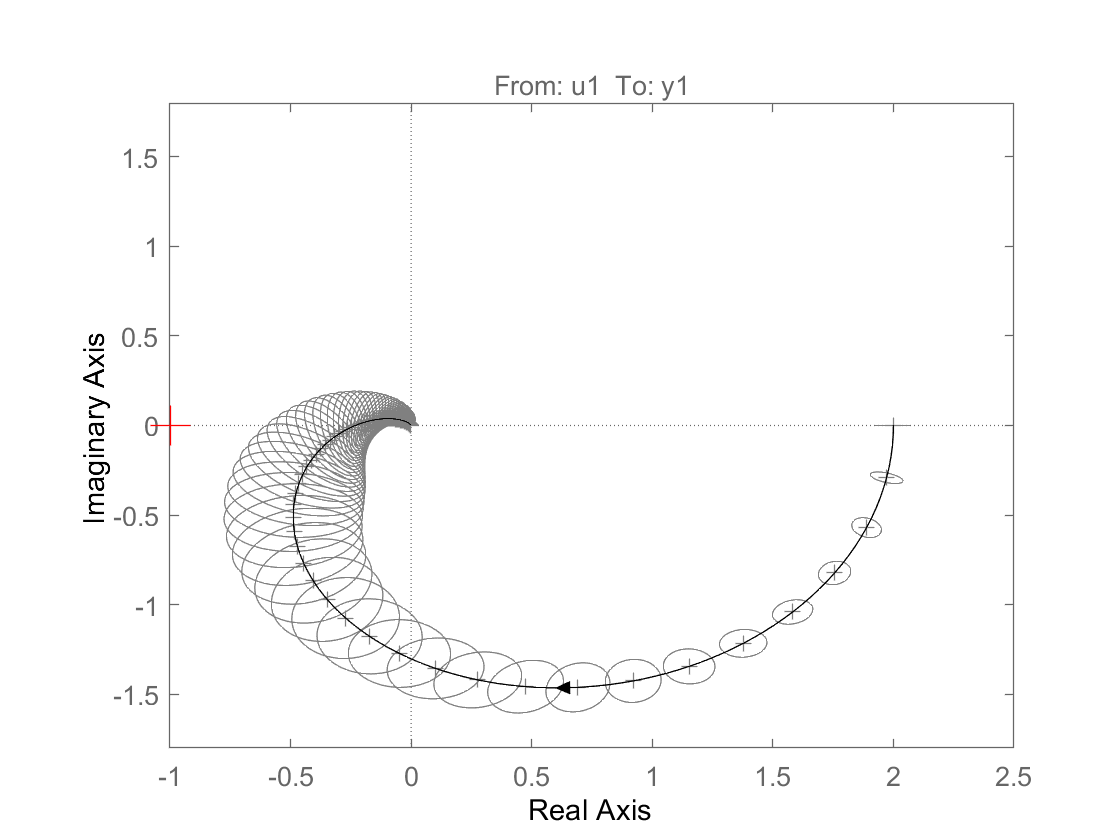}
    \caption{Nyquist plot of the identified system \eqref{eq:model_set} along with uncertainty regions corresponding to $5$ standard deviations on a frequency grid of interval length $5\,\mathrm{rad/s}$ is shown here.}
    \label{fig:nyq_plt}
\end{figure}

As a preliminary step towards addressing this question,
we study the random variable $\kappa(P(j\omega), \bar{P}(j\omega))$ at given frequency $\omega\in\mathbb{R}$, 
for given nominal model $\bar{P}$; we often drop the argument $j \omega$ subsequently for brevity of notations.
Given $d \in [0,1]$, the corresponding CDF 
is given by
\begin{align*}
\mathbb{P}
\left(
\kappa(P, \bar{P}) \leq d
\right) 
&= 
\mathbb{P}
\left(
\kappa(P, \bar{P}) - \rho(\bar{P},\bar{C}) \leq d - \rho(\bar{P},\bar{C})
\right) \\
&= 
\mathbb{P}
\left(
\rho(\bar{P},\bar{C}) - \kappa(P, \bar{P}) \geq \rho(\bar{P},\bar{C}) - d 
\right) \\
&\leq
\mathbb{P}
\left(
\rho(P,\bar{C}) \geq \rho(\bar{P},\bar{C}) - d
\right),
\end{align*}
where the last inequality follows from \eqref{eqn_bpc_performance_degrade_pointwise}. In particular, the value of the CDF of $\kappa(P,\bar{P})$ at the specified value $d$ is a controller independent lower bound for $\mathbb{P}\left(
\rho(P,\bar{C}) \geq \rho(\bar{P},\bar{C}) - d
\right)$. Understanding this quantity is a precursor to understanding $\mathbb{P}(b_{P,\bar{C}} \geq b_{\bar{P},\bar{C}} - d)$, which is the probability that closed-loop performance degrades by at most $d$ for any controller that stabilizes $\bar{P}$. The alternative is to more directly study the random variable $\rho(P,\bar{C})$ for a particular choice of controller $\bar{C}$, limiting the scope for addressing Question 1.
With this in mind, the main aim in this paper is to characterise the relationship between the distribution of the random variable $\kappa(P(j\omega),\bar{P}(j\omega))$, and an assumed known distribution of $P(j\omega)$ in the complex plane, at a given frequency $\omega$. 
We present below two motivating problems where closed-loop performance degradation with respect to random plants may arise.

\subsection{Motivating Problem $1$}
\begin{figure}
    \centering    \includegraphics[width=\linewidth]{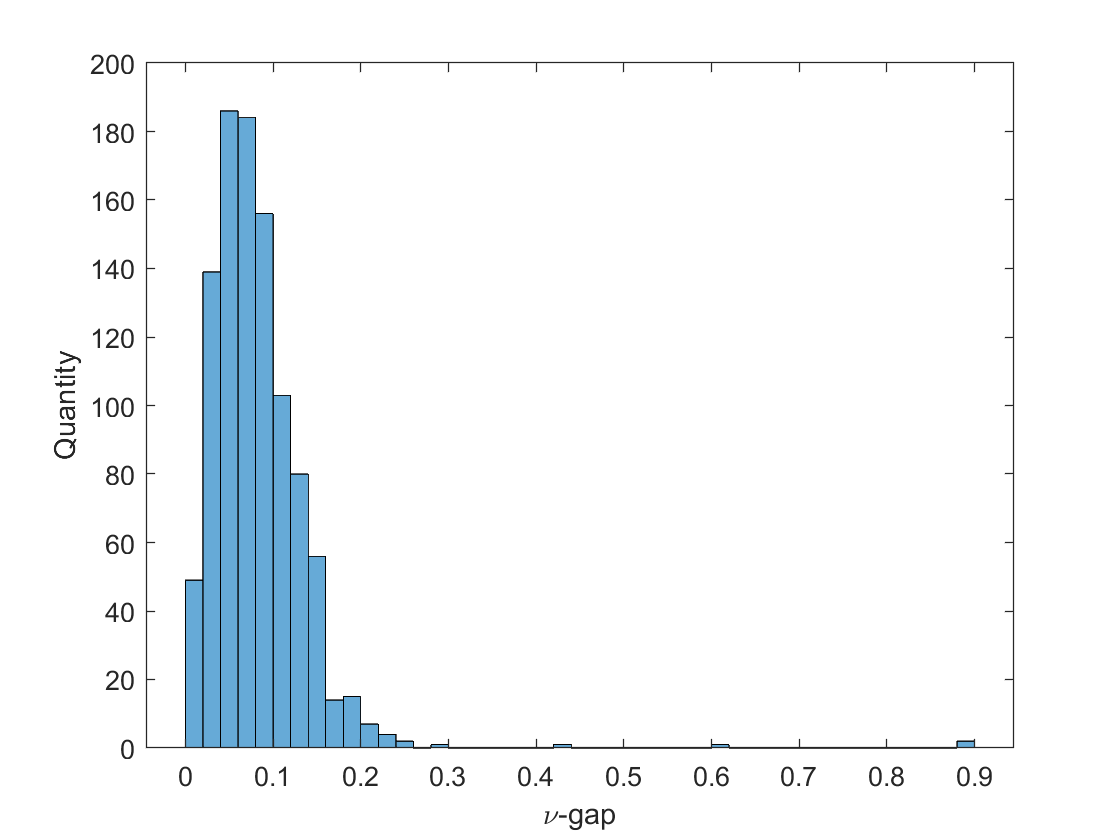}
    \caption{The histogram of the $\nu$-gap between the nominal model and $1000$ independent trials is shown here. Clearly, more systems with very small distance measure from the nominal model got realised during the trials.}
    \label{fig:nugap-histogram}
\end{figure}
Suppose that an LTI model with no zeros and three poles in the same location models the process well. Then, the transfer function corresponding to that process will be of the form 
\begin{equation}\label{eq:model_set}
    P(s)=\frac{b}{(1+\tau s)^3}.
\end{equation}
Here $b,\tau\in\mathbb{R}$ are \emph{unknown} parameters and may be estimated from data. Since all physical system outputs are subject to measurement noise, uncertainty is introduced into the identification process. In Figure \ref{fig:nyq_plt}, one such estimation of the example process Nyquist plot is shown where a random binary sequence was generated as the input to a chosen nominal model $\bar{P}(s)$ (obtained using $b = 2$, $\tau = \frac{1}{10}$ in \eqref{eq:model_set}) and Gaussian white noise was added to the output representing measurement noise before estimating the parameters. In the figure we can also see point-wise confidence regions (each of which corresponds to $5$ standard deviations) on a frequency grid of interval length $5\,\mathrm{rad/s}$. 

For every frequency, a region around the estimated point now corresponds to a region where the Nyquist plot of the true system at that frequency is likely to be, given that the model parametrization accurately reflects the true system. Note that for this example, the gain uncertainty at the phase crossover frequency suggests that if we had identified the model from a different data sample, we could very well identify a model with a very different gain margin for a stabilizing controller. 
Though we do not know what the true system is, we believe that it belongs to some set of models in proximity to the estimated nominal model with respect to their Nyquist plots, and by extension their closed loop behaviour. Further, we have an estimation of the probability distribution over that set of models by estimating parameter uncertainty.

\begin{figure}
    \centering
    \includegraphics[width=\linewidth]{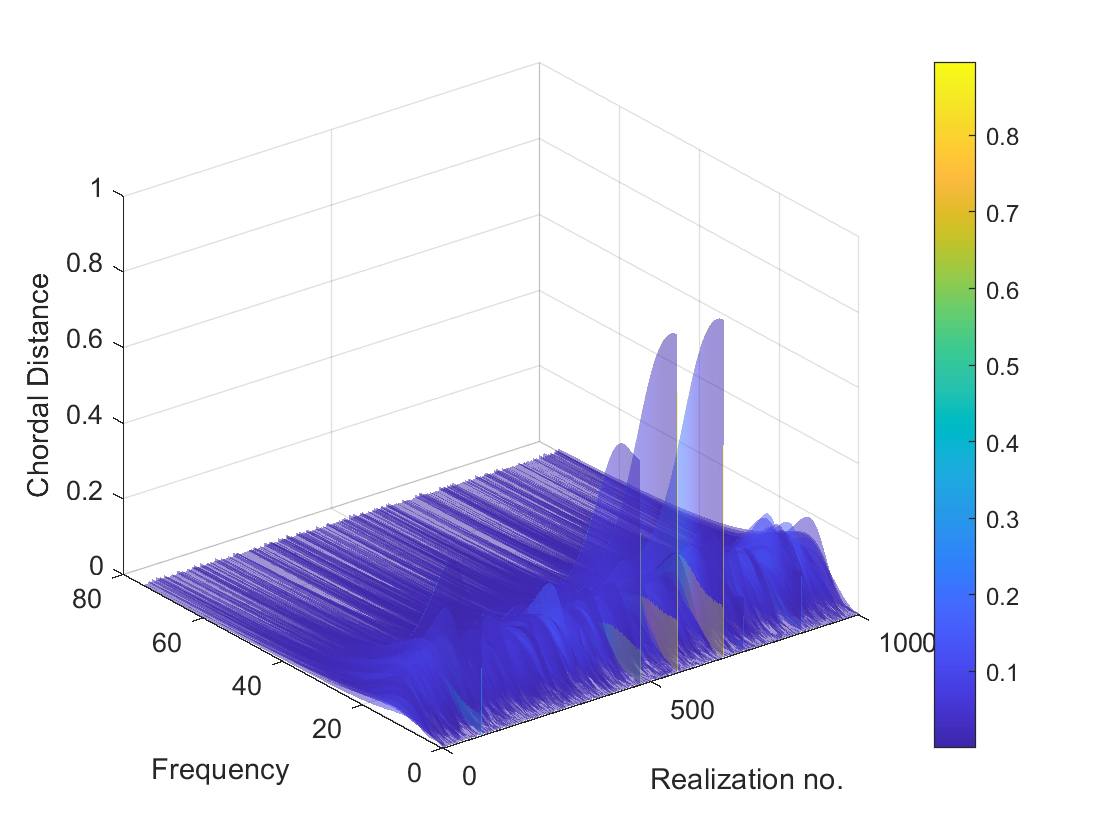}
    \caption{Realizations of chordal distance from the nominal model across a range of frequencies during $1000$ trials are plotted here.}
    \label{fig:realization-surf}
\end{figure}

To motivate further, we simulated the system identification procedure described earlier for $N = 1000$ independent trials (with different random binary sequences inputs and measurement noise) to identify $N$ plants of the form \eqref{eq:model_set}. The results shown in Figures \ref{fig:nugap-histogram} and \ref{fig:realization-surf} indicate that the point-wise chordal distance can be formulated as a random quantity. 

\subsection{Motivating Problem $2$}
\begin{figure}
    \centering
    \begin{tikzpicture}

\def\Radius{2cm} 
\coordinate (O) at (0,0); 
\coordinate (N) at (0,\Radius); 
\coordinate (S) at (0,-\Radius); 

\draw[color=brown!100,line width=2pt] (O) circle (\Radius);

\coordinate (P) at (-10:\Radius); 
\coordinate (Q) at (-40:\Radius); 
\coordinate (R) at (-25:\Radius); 
\coordinate (A) at (1.9,-\Radius); 
\coordinate (B) at (3.3,-\Radius); 
\coordinate (C) at (2.6,-\Radius); 
\draw[fill=red] (C) circle (1.5pt) node[above right] {};

\draw[fill=blue] (R) circle (1.5pt) node[right] {$\phi^{-1}(\bar{P})$};
\draw[fill=red] (Q) circle (1.5pt) node[left] {$\phi^{-1}(P)$};

\draw[dashed, color=black] (S) -- (N);
\draw[dashed, color=blue, line width=1.25pt] (R) -- (N);
\draw[solid, color=blue, line width=1.25pt] (R) -- (C);
\draw[dashed, color=red, line width=1.25pt] (N) -- (Q);
\draw[solid, color=red, line width=1.25pt] (Q) -- (A);
\draw[solid, color=green, line width=2.5pt] (R) -- (1.50,-1.3);


\draw[thick] (-4,-\Radius) -- (4,-\Radius);
\draw[thick, color=orange, line width=3pt] (A) -- (B);

\node at (S) [below] {$\mathfrak{S}$};
\node at (N) [above] {$\mathfrak{N}$};
\node at (C) [above right] {$\bar{P}$};
\node at (A) [above left] {$P$};
\node at (-3,-\Radius) [above] {$\mathbb{R}$};
\node at (B) [below left] {$\mathbf{S}$};
\draw[color=orange, line width=1pt,fill=orange!40, opacity=0.2] (N) -- (A) -- (B) -- (N);
\draw pic["$\alpha$",draw=black,->,angle eccentricity=1.2,angle radius=1.75cm] {angle=Q--N--R};

\end{tikzpicture}
\caption{Stereographic projection on $\mathbb{R}$ is illustrated here. The uncertainty in point $P$ is depicted as an orange interval $\mathbf{S} \subset \mathbb{R}$. Both the nominal point $\bar{P}$ and its projection $\phi^{-1}(\bar{P})$ are shown in blue color. A realization of the random point $P \in \mathbf{S}$ and its projection $\phi^{-1}(P)$ are shown in red color. Since $P$ is random, the angle $\alpha := \measuredangle P\mathfrak{N}\bar{P}$ comes random and hence the corresponding chordal distance line in green color are random as well.}
\label{fig_nu_gap_real_case}
\end{figure}
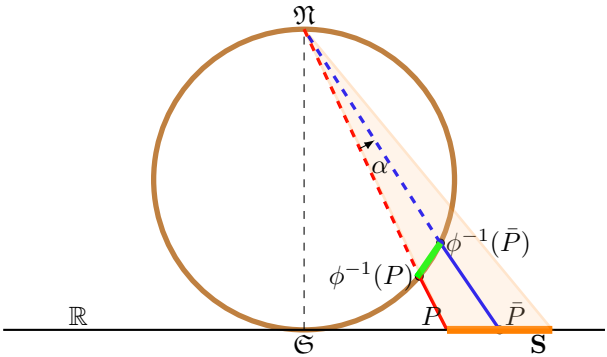
Let us consider the simple scalar real case $(P \in \mathbb{R})$. Now consider a circle $\mathcal{C}$ with unit diameter tangent with respect to real number line at its south pole ($\mathcal{C}$ is the analogue of Riemann sphere in the complex setting which will be defined later in the paper). The line that connects the north pole of the circle and the points on $\mathbb{R}$ intersects the circle precisely at one point and the mapping $\phi: \mathcal{C} \rightarrow \mathbb{R}$ is called the Stereographic Projection. Now consider the case when the point $P \in \mathbb{R}$ is random and let $\mathbf{S} \subset \mathbb{R}$ denote the compact support set of the distribution $f_{\mathsf{P}}$ governing the uncertainty in $P$. Further, assume that a nominal value denoted by $\bar{P} \in \mathbf{S} \subset \mathbb{R}$ is known apriori. Note that, $P \sim f_{\mathsf{P}}$ can take any value in $\mathbf{S}$. The distance between the projected values of nominal value $\bar{P}$ and any realization $P$ on the circle $\mathcal{C}$ becomes random as well. A simple illustration is shown in Figure \ref{fig_nu_gap_real_case}. Note that the angle $\alpha := \measuredangle P\mathfrak{N}\bar{P}$ is a random quantity and hence the chordal distance between $\phi^{-1}(\bar{P})$ and $\phi^{-1}(P)$ given by $|\sin(\alpha)|$ turns out to be random as well. In such a setting, we want to study the expected amount of perturbation needed from $\bar{P}$ to realize any $P \in \mathbf{S}$. 



\subsection{Contribution and Paper Organization} In this paper, our main contribution is a formulation of the chordal distance between the stereographic projections of Nyquist plots at a given frequency as a random quantity, and the derivation of an explicit formula for its CDF; see Theorem \ref{thm_pdf_kappa}. To answer Question \ref{question_1}, we need the CDF of the $\nu$-Gap, which requires us to deal with the correlation of Nyquist plot along frequencies, apart from handling the winding number condition. We are pursuing this significant challenge in ongoing work.

The rest of the paper is organised as follows: Following a short introduction of the preliminaries, the Nyquist plots with probabilistic uncertainty and their stereographic projection from the Riemann sphere is presented in \S\ref{sec_main_SISO} along with the discussion of the CDF of the chordal distance point-wise in frequency. The proposed idea is demonstrated using a numerical example in \S\ref{sec_num_sim}. The results presented in this paper can be reproduced using the code that is made publicly available at \url{https://github.com/antonnystrm/CDC_2024}. Finally, the main findings of the paper are summarized in \S\ref{sec_conclusion} along with a planned directions for the future research. For more comprehensive details on the ideas presented in this paper, the interested readers are referred to \cite{Anton_MS_Thesis}.   

\subsection{Notation and Preliminaries}
The cardinality of the set $A$ is denoted by $\left | A \right \vert$. The operator $\backslash$ denotes the set subtraction.
The set of real numbers, integers and the natural numbers are denoted by $\bbr, \bbz, \bbn$ respectively. 
The space of complex numbers is denoted by $\mathbb{C}$ and $j$ represents the imaginary unit. The real and imaginary parts of the complex number $z \in \mathbb{C}$ is denoted by $\mathrm{Re}(z)$ and $\mathrm{Im}(z)$ respectively. Given $z \in \mathbb{C}$, we denote its complex conjugate as $z^{\star} \in \mathbb{C}$.
 The subset of real numbers greater than $a \in \bbr$ is denoted by $\bbr_{> a}$. 
For convenience, we will write parameterized multi-variable functions as follows $f(x(t),y(t))=f(x,y)(t)$.

A probability space is defined by a triplet $(\Omega, \mathcal{F}, \mathbb{P})$, where $\Omega$, $\mathcal{F}\subset 2^\Omega$, and $\mathbb{P}$ denote the sample space, event space and the probability function respectively with $\mathbb{P}: \mathcal{F} \rightarrow [0,1]$. A real random variable $x \in \mathbb{R}$ with probability density function $f_{\mathsf{x}}$ is denoted by $x \sim f_{\mathsf{x}}$ and the CDF is denoted by $\mathbb{F}_{\mathsf{x}}(X) = \mathbb{P}(x \leq X) = \int_{-\infty}^{X} f_{\mathsf{x}}(x)\mathrm{d}x$.
A complex random variable $z$ on the probability space $(\Omega, \mathcal{F}, \mathbb{P})$ can be considered as a random vector in $\mathbb{R}^{2}$ with the real and imaginary parts being real random variables in each dimension.

\section{Stereographic Projection} \label{sec_main_SISO}
In this section, we first analyse the uncertainty at a point on the Nyquist plot as described by a distribution in $\mathbb{C}$ with a convex, compact support set. In particular, we would like to characterize the uncertainty associated with the
chordal distance between corresponding stereographic projections onto the Riemann sphere relative to a nominal model. 

Given a frequency $\omega$, consider the complex value of the transfer function of the plant $P(j \omega) \in \mathbb{C}$ as uncertain. We take a stochastic approach in which $P(j \omega) \sim f_{\mathsf{P}}$, meaning $P(j \omega)$ is a random variable in the space 
$(\Omega(\omega), \mathcal{F}(\omega), \mathbb{P}(\omega))$ whereby $P: \Omega(\omega) \rightarrow \mathbf{S}(\omega) \subset \mathbb{C}$. It is assumed that the support of $f_{\mathsf{P}}$ is the compact set $\mathbf{S}(\omega)$.\footnote{Without compact $\mathbf{S}$, there will be a non-zero probability that the uncertainty around the Nyquist plot at every frequency includes the $-1$ point, precluding unit gain feedback stability across the set of uncertain plants.}
That is, the probability density function $f_{\mathsf{P}}$ is a non-negative mapping that integrates to $1$ over $\mathbf{S}(\omega)$. 
Furthermore, we assume that the nominal model $\bar{P}(j \omega) \in \mathbf{S}(\omega)$ has been fixed. From now on, we drop the frequency argument for convenience, on the understanding that all the following developments corresponds to the quantities at the particular frequency $\omega$. 

The Riemann sphere, denoted by $\mathfrak{R}$, models the extended complex plane. It is the surface of the unit diameter sphere in $\mathbb{R}^3$, tangent to $\mathbb{C}\sim\mathbb{R}^2$ at the origin. 
Its north and south pole are denoted by $\mathfrak{N}=(0,0,1)$ and $\mathfrak{S}=(0,0,0)$, respectively.
\begin{definition}
Let $R = (x, y, z) \in \mathfrak{R} \backslash \mathfrak{N}$. The straight line through $\mathfrak{N}$ and $R$ intersects $\mathbb{C}$
exactly at one point $\phi(R) \in \mathbb{C}$. 
The map $\phi: \mathfrak{R} \backslash \mathfrak{N} \rightarrow \mathbb{C}$
given by $R \mapsto \phi(R)$ is called the stereographic projection.
\end{definition}

Since the distribution $f_{\mathsf{P}}$ with a convex, compact support set $\mathbf{S}$ is defined on $\mathbb{C}$, we require the inverse of the stereographic projection. Let $\phi(R) \in \mathbb{C}$ 
be a point on the complex plane.
Then, the straight line from $\phi(R) \in \mathbb{C}$ to $\mathfrak{N}$ intersects the boundary of the Riemann sphere exactly at one point $R = (x, y, z) \in \mathfrak{R} \backslash \mathfrak{N}$. The map $\phi^{-1} : \mathbb{C} 
\rightarrow \mathfrak{R} \backslash \mathfrak{N}$ given by $\phi(R) \mapsto R$ is the required inverse of the stereographic projection. Given this, we would like to characterize the set $\mathfrak{R}_{\mathbf{S}} \subset \mathfrak{R}$. We define the \emph{projected support set} 
by
\begin{align}
\mathfrak{R}_{\mathbf{S}} 
&= \phi^{-1}(\mathbf{S}) 
:=
\left\{
R \in \mathfrak{R} 
\mid
R = \phi^{-1}(r),~ r \in \mathbf{S}
\right\}.
\end{align}

Given $c \in \mathbb{C}$,
the corresponding point $R = (r_x, r_y, r_z) \in \mathfrak{R}\backslash\mathfrak{N}$ given by $R = \phi^{-1}(c)$ has the following Cartesian coordinates:
\begin{align} \label{eqn_sphere_cartesian_coords}
r_x = \frac{\mathrm{Re}(c)}{1 + \absval{c}^{2}}; ~ 
r_y = \frac{\mathrm{Im}(c)}{1 + \absval{c}^{2}}; ~
r_z = \frac{\absval{c}^{2}}{1 + \absval{c}^{2}}.
\end{align}
Similarly, given a point $R = (r_x, r_y, r_z) \in \mathfrak{R} \backslash \mathfrak{N}$, the coordinates of the corresponding stereographically projected point $c \in \mathbb{C}$ is given by
\begin{align}
\label{eqn_complexplane_cartesian_coords}
    \phi(R) 
    =
    \left(\frac{r_x}{1 - r_z}, \frac{r_y}{1 - r_z}\right) =
    \frac{r_x}{1 - r_z} + j  \frac{r_y}{1 - r_z}.
\end{align}

\subsection{Formulating Chordal Distance as a Random Quantity}

The nominal model $\bar{P} \in \mathbf{S}  \subset \mathbb{C}$ can be mapped to a point $\bar{R} \in \mathfrak{R}_{\mathbf{S}} \subseteq \mathfrak{R}$ using the map $\phi^{-1}$. Since the stereographic projection constitutes a bijection between the sets $\mathfrak{R}_{\mathbf{S}}$ and $\mathbf{S}$, we may infer a probability density over $\mathfrak{R}_{\mathbf{S}}$ from the density $f_{\mathsf{P}}$ over $\mathbf{S}$. 
The chordal distance 
between $R$ and $\bar{R}$ is indeed a random variable like $P$ in the probability space 
$(\Omega, \mathcal{F}, \mathbb{P})$, meaning $K := \kappa(\bar{P}, P): \Omega \rightarrow [0,1]$. We now state the main problem of interest of this paper. 

\begin{problem} \label{problem_1}
Determine the CDF of $K$ given $f_{\mathsf{P}}$ with convex compact support $\mathbf{S}$ and the nominal model $\bar{P}$.    
\end{problem}

Two perspectives on this problem are considered in the rest of this section. The second takes a form that may facilitate extension of the results to MIMO systems. The challenges associated with extending these single frequency chordal distance results to corresponding $\nu$-gap results is then discussed. The following lemma is used to establish the first perspective.
\begin{lemma} \label{lemma_fzw}
Given random variables $\mathsf{x}$ and $\mathsf{y}$, with joint density $f_{\mathsf{xy}}(x,y)$, and constants $a,b\in\mathbb{R}$, define $r:=\sqrt{a^2+b^2}$, and the random variables
\begin{subequations}
\label{eqn_zw}
\begin{align}
    \mathsf{z} &:= g_{1}(\mathsf{x},\mathsf{y}) = (\mathsf{x}-a)^{2} + (\mathsf{y}-b)^{2}, \\
    \mathsf{w} &:= g_{2}(\mathsf{x},\mathsf{y}) = \mathsf{x}^{2} + \mathsf{y}^{2} + 1.
\end{align}
\end{subequations}
Then, the joint density of $\mathsf{z}$ and $\mathsf{w}$ is given by
\begin{align}\label{eqn_fzw}
    f_{\mathsf{zw}}(z,w)=\begin{cases}
        f(z,w),&\vert r-\sqrt{w-1} \vert<\sqrt{z}\\
        \text{undefined}, &\vert r-\sqrt{w-1} \vert=\sqrt{z}\\
        0, &\text{otherwise}
    \end{cases}
\end{align}
where 
\begin{subequations}
\begin{align}
f(z,w)
&:= 
\frac{1}{4 r^{2}c_2(z,w)}\sum^{2}_{i = 1} f_{\mathsf{xy}}(x_{i}(z,w), y_{i}(z,w)) \\
c_1(z,w) 
&:= 
\left(\frac{1}{2} - \frac{z-w+1}{2 r^{2}}\right) \\
c_2(z,w) 
&:= 
\frac{1}{2} \sqrt{\frac{2(z+w-1)}{r^{2}} - \frac{(z-w+1)^{2}}{r^{4}} - 1}
\end{align}
\end{subequations} 
and
\begin{align}\label{eqn_intersections}
\begin{bmatrix}
x_i \\ y_i    
\end{bmatrix}    
&= 
c_1(z,w)
\begin{bmatrix}
a \\ b    
\end{bmatrix} + (-1)^{i-1} 
c_2(z,w)
\begin{bmatrix}
-b \\ a    
\end{bmatrix},\, i=1,2
\end{align} 
denote the solutions to \eqref{eqn_zw} for $(z, w)$ satisfying $\left\lvert r-\sqrt{w-1} \right\rvert<\sqrt{z}$. 
\end{lemma}
\begin{proof}
The proof is deferred to the appendix.
\end{proof}
\noindent \textbf{Remark 1:} Note that \eqref{eqn_fzw} is a valid probability density function.

\noindent From \eqref{eqn_kappa_siso}, $K$ can be written as
\begin{align}
K
&:= 
\frac{\left| \bar{P} - P \right\vert}{\sqrt{1 + \left| \bar{P} \right\vert^{2}} \sqrt{1 + \left| P \right\vert^{2}}} = 
\frac{1}{c} 
\sqrt{\frac{g_{1}(\mathsf{x},\mathsf{y})}{g_{2}(\mathsf{x},\mathsf{y})}} \label{eqn_kappa_variables},
\end{align}
where $c = \sqrt{1 + r^{2}}$ is known constant quantity with $r = \left| \bar{P} \right\vert$. Now let $Q = \frac{\mathsf{z}}{\mathsf{w}}$ be a random variable constructed using \eqref{eqn_zw} from Lemma \ref{lemma_fzw}. Then, from \eqref{eqn_kappa_variables} it follows that 
\begin{align} \label{eqn_K_qc}
    K = \frac{\sqrt{Q}}{c}.
\end{align}
It is now possible to estalish a first answer to problem \ref{problem_1}.
\begin{theorem} \label{thm_pdf_kappa}
Let $P =x+jy \in \mathbb{C}$ follow the distribution $f_{\mathsf{P}}$ with compact support $\mathbf{S}$ corresponding to the joint density $f_\mathsf{xy}$ of the real random variables $x$ and $y$. Further, assume that nominal model $\bar{P} \in \mathbb{C}$ is known.
Then, for $d\in[0,1]$, the CDF of $K:=\kappa(P,\bar{P})$ is
\begin{equation}
\label{eqn_kappa_pdf}
\mathbb{F}_{\mathsf{K}}(d) 
=
\int^{\infty}_{1} \int^{u(d,t)}_{\ell(d,t)}
\frac{1}{4r^2c_2(l,t)}
\sum^{2}_{i = 1} f_{\mathsf{xy}}(x_{i}, y_{i})(l,t) \, dl \, dt,
\end{equation}
where $r = \left| \bar{P}\right\vert$ and
\begin{align*}
u(d,t)
&= \min\{ td^{2}(1+r^{2}),(r + \sqrt{t-1})^{2}\}
\\
\ell(d,t)
&= \min\{ td^{2}(1+r^{2}),(r - \sqrt{t-1})^{2}\}\\
c_2(l,t) 
&= 
\frac{1}{2} \sqrt{\frac{2(l+t-1)}{r^{2}} - \frac{(l-t+1)^{2}}{r^{4}} - 1}.    
\end{align*}
\end{theorem}

\begin{proof}
Let the known nominal model be $\bar{P} = a+ bj$.
For a given $d \in [0,1]$, using Lemma \ref{lemma_fzw} and \eqref{eqn_K_qc}, the CDF of $K$ is given by 
\begin{align*}
\mathbb{F}_{\mathsf{K}}(d)
= \mathbb{P}\left(\frac{\sqrt{Q}}{c} < d\right) 
= \mathbb{P}\left( Q < c^2 d^2 \right) 
= \mathbb{F}_{\mathsf{Q}}(c^2 d^2),
\end{align*}
where $\mathbb{F}_{\mathsf{Q}}(\cdot)$ denotes the CDF of $Q$ and is described \cite{Papoulis_Book} as
\begin{align*}
    \mathbb{F}_{\mathsf{Q}}(q)
    :=
    \int^{\infty}_{0} \int^{tq}_{0}
    f_{\mathsf{zw}}(l,t) \, dl \, dt.
\end{align*}
Here, $f_{\mathsf{zw}}(l,t)$ denotes the joint probability distribution function (PDF) of $z$ and $w$ obtained using Lemma \ref{lemma_fzw}. By realizing that $f_{\mathsf{zw}}(z,w)$ is zero outside the region specified by $\left\lvert r-\sqrt{w-1} \right\rvert < \sqrt{z}$, we get 
\begin{align*}
    \mathbb{F}_{\mathsf{Q}}(q)
    &=
    \int^{\infty}_{1} \int^{\min\{tq,\left(r+\sqrt{t-1}\right)^2\}}_{\min\{tq,\left(r-\sqrt{t-1}\right)^2\}}
    f(l,t) \, dl \, dt, \quad \text{where} \\  
    f(l,t)
    &=
    \frac{1}{4 r^{2}c_2(l,t)}\sum^{2}_{i = 1} f_{\mathsf{xy}}(x_{i}(l,t), y_{i}(l,t)).
\end{align*}
Here, $x_i(l,t), y_i(l,t)$ are given by \eqref{eqn_intersections}.
Setting $q = c^{2} d^{2} = (1 + r^{2}) d^{2}$ yields the final result.
\end{proof}


In general, to determine the CDF, we need to describe the area under which the integration has to be carried out given a distance threshold. To this end, note that the open ball of radius $d \in (0,1)$ centered on $\bar{R} = \phi^{-1}(\overline{P})$ on the Riemann sphere is given by
\begin{align} \label{eqn_open_ball}
\mathcal{B}_d\left(\bar{R}\right)
:=
\left\{
R \in \mathfrak{R}
\mid
\kappa(\phi(\bar{R}),\phi(R)) < d
\right\} \subset \mathfrak{R}.
\end{align}
The set $\phi\left(\mathcal{B}_d\left(\bar{R}\right)\right) \subset \mathbb{C}$ is again a ball in $\mathbb{C}$, albeit with a center-point that is different to $\bar{P}$, and radius $\hat{d}$ that is possibly not equal to $d$. The next lemma determines both $\hat{d}$ and $\hat{P}$ as a function of $d$ and $\bar{R}$.
\begin{lemma}
\label{lemma_complex_ball_radius}
Let $\bar{R}=(x,y,z)\in \mathfrak{R}$ be given. Given an open ball $\mathcal{B}_{\frac{d}{2} \left(\cos{\frac{\Delta\theta}{4}}\right)^{-1}}\left(\bar{R}\right)\subset \mathfrak{R}$, with $d \in [0,1]$, such that $\mathcal{B}_{\frac{d}{2} \left(\cos{\frac{\Delta\theta}{4}}\right)^{-1}}\left(\bar{R}\right) \cap \mathfrak{N} = \emptyset$, the Euclidean diameter of its stereographic projection $D := \phi\left(\mathcal{B}_{\frac{d}{2} \left(\cos{\frac{\Delta\theta}{4}}\right)^{-1}}\left(\bar{R}\right) \cap \mathfrak{R}\right) \subset \mathbb{C}$ is given by 
\begin{subequations}
\begin{align}
\hat{d}
&=\left\lvert \frac{\sin(\theta+\Delta\theta)}{1-\cos(\theta+\Delta\theta)}- \frac{\sin\theta}{1-\cos\theta}\right\rvert, \quad \text{where}, \\ 
\sin{\frac{\Delta\theta}{2}}&
=d, \theta = \theta(\bar{R}) \text{ such that }
\cos{\left(\theta + \frac{\Delta\theta}{2}\right)} = 2z-1.
\end{align}
\end{subequations}
Further, the center point $\hat{P}\in \mathbb C$ of the disc $D$ is given by
\begin{subequations}
    \begin{align}
        &\hat{P}=\frac{\left\lvert \frac{\sin{\left(\theta+\Delta\theta\right)}}{1-\cos{\left(\theta+\Delta\theta\right)}}\right\rvert+\left\lvert\frac{\sin\theta}{1-\cos\theta}\right\rvert}{2}\left(
            \cos\varphi(\Bar{R}) + \sin\varphi(\Bar{R}) j
        \right), \\ 
        &\text{where} \nonumber \\
        &\varphi(\Bar{R}) 
        \colon= 
        \mathrm{sign}(y)\arccos{\left(\frac{x}{\sqrt{x^2+y^2}}\right)}.\label{eq:phi(R)}
    \end{align}
\end{subequations}
\end{lemma}
\begin{proof}
The proof is deferred to the appendix.    
\end{proof}
\inArxiv{
\textbf{Remark:} 
Note that for a point $\Bar{R}$ close to the north pole then $\theta$ is small which makes $\hat{d}$ large for any $d$ and for $\Bar{R}$ close to the south pole then $\theta$ is close to $\pi$ making $\hat{d}$ very small for any $d$.
}
\noindent To precisely define the CDF of $K$, we define
\begin{align}
\label{eqn_spd_region}
    \mathbf{S}^{\hat{P}}_{d} 
    :=     \phi\left(\mathcal{B}_d\left(\bar{R}\right)\right) \, \cap \, \mathbf{S} .
\end{align}
Then, for $d \in [0,1]$, the CDF of the point-wise chordal distance between the random quantity $P$ and the known nominal model $\bar{P}$ 
is given by
\begin{align}
\label{eq:cdf pcd general}
\mathbb{F}_{\mathsf{K}}(d)
&:=
\int^{d}_{0} f_{\mathsf{K}}(k) \, dk \nonumber = \mathbb{P}\left(\kappa\left(\Bar{P},P\right)<d\right)\\
&= \mathbb{P}
\left(P\in {\mathbf{S}^{\hat{P}}_{d}}\right) 
=
\int_{\mathbf{S}^{\hat{P}}_{d}} f_{\mathsf{P}}(a) \, da.
\end{align} 

\textbf{Remark 2:} 
The aforementioned second perspective follows from the expression for $\mathbb{F}_{\mathsf{K}}$ in terms of $f_{\mathsf{P}}$ and the integral in \eqref{eq:cdf pcd general}. Indeed, it is evident that $\mathbb{F}_{\mathsf{K}}$ can be computed using the Euclidean distance in $\mathbb{C}$ rather than the chordal distance metric on the Riemann sphere. Note that the expression for the CDF given by \eqref{eq:cdf pcd general} is in a more general form, as it can handle any support set $\mathbf{S}$ and any distribution $f_{\mathsf{P}}$. It provides scope for generalizing to MIMO systems $\bar{P}$. 

\textbf{Challenges:} The next step to extend this analysis to find the CDF of the $\nu$-gap, requires various challenges to be overcome. While it is possible to define the family of random variables $K_\omega := \kappa(\bar{P}(j\omega), P(j\omega))$,
the instances $K_{\omega_{1}}$ and $K_{\omega_{2}}$ at two different frequencies would be correlated. Dealing with this has inherent challenges given that finding a joint density of $K_{\omega_{1}}$ and $K_{\omega_{2}}$ 
would require more knowledge about the source of uncertainty in underlying random transfer function $P$; one possibility to consider is use of Copula computed via Sklar's theorem \cite{sklar1959fonctions}. 
Further, determining $\mathbb{P}( (\forall \omega)~ K_{\omega} < \beta)$ for given $\beta \in [0,1]$ involves conjunction across all frequencies, as a further complication. 
On the other hand, we could instead use disjunction and characterization of the existence of a particular frequency, say $\hat{\omega}$ where $\mathbb{P}(K_{\hat{\omega}} > d)$, to establish a conservative upper bound without the need to establish joint densities across all frequency. Consideration of $K_\omega$ as a stochastic process is yet another possible approach.
\section{Numerical Simulations}
\label{sec_num_sim}
To further provide anchoring to the result given in Theorem \ref{thm_pdf_kappa}, this section provides a numerical example where the integral \eqref{eqn_kappa_pdf} is plotted over $d \in [0,1]$ to show the resulting CDF. In the example, we have selected the distribution $f_{\mathsf{P}}$ to be Gaussian on $\mathbb{C}$ and centered on $\bar{P}=1+j$ with covariance matrix given by $\mathrm{diag}(1, \frac{1}{4})$. The numerical integration was performed using Matlab's $\mathrm{integral2}$ function. In Figure \ref{fig:KwCDF}, the valid monotone CDF, $\mathbb{F}_{\mathsf{K}}$ can be seen. 
\begin{figure}
    \centering
    \includegraphics[width=\linewidth]{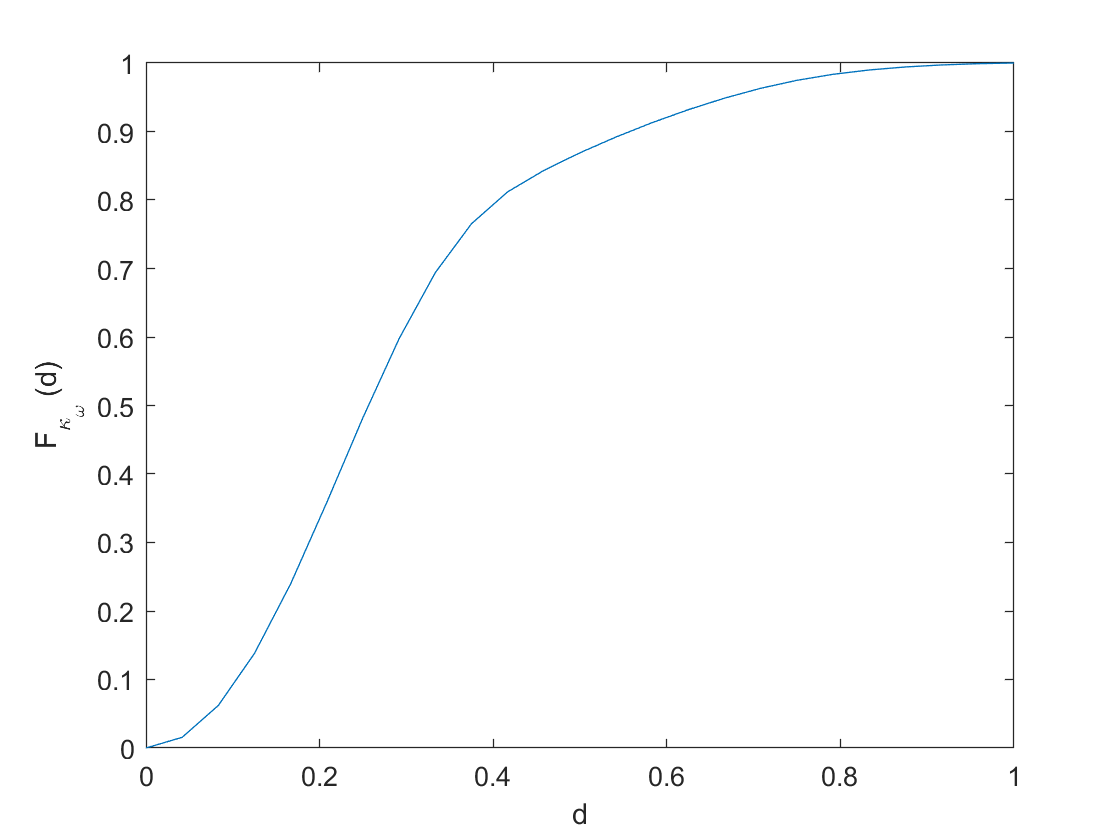}
    \caption{The CDF of $\kappa(\bar{P}, P)$ following Theorem \ref{thm_pdf_kappa} when the underlying distribution $f_{\mathsf{P}}$ is Gaussian is depicted here.}
    \label{fig:KwCDF}
\end{figure}

\section{Conclusion \& Future Outlook} 
\label{sec_conclusion}
A fresh perspective on probabilistic robust control theory by treating the point-wise chordal distance metric between LTI dynamical systems as a random quantity was presented. Knowing the CDF of the plant model at a frequency, the corresponding CDF of the point-wise chordal distance metric was obtained. The proposed framework presents a fertile research landscape for future research. A systematic development of future research articles is planned in the following sequence to concretize the proposed theory properly involving both SISO and MIMO extensions.  
\begin{enumerate}
    \item Extend the analysis from point-wise chordal distance setting made in this paper to $\nu$-Gap setting.
    \item Given two LTI systems $P_1, P_2$ following respective distributions $\mathbb{P}_1, \mathbb{P}_2$ with compact supports $\mathbf{S}_1, \mathbf{S}_2$ at a frequency $\omega \in \mathbb{R}_{\geq 0}$, find the worst case expectation of the point-wise chordal distance between their projected points onto the Riemann sphere exceeding a threshold. It would be one of the first attempts at computing the distance between two stochastic dynamical systems. 
\end{enumerate}


\inConf{
\bibliographystyle{IEEEtran}
\bibliography{references}
}
\inArxiv{
\bibliographystyle{tmlr}
\bibliography{references}
}

\inArxiv{
\section*{Appendix}
}
\appendix
\subsection*{Proof of Lemma \ref{lemma_fzw}}
We will use the well known result (e.g., see \cite{Papoulis_Book, Geometric_Measure_Book}) that for a mapping of the kind in \eqref{eqn_zw}, with more general functions $g_1$ and $g_2$, the following holds true:
\begin{align}\label{eqn_cdf_transform}
    f_{\mathsf{zw}}(z,w)=\sum_{i=0}^n\frac{1}{\vert J(x_i,y_i)\vert}f_{\mathsf{xy}}(x_i,y_i)
\end{align}
where $x_i,y_i$ are all solutions to the mapping for the corresponding $z,w$ and $J(x_i,y_i)$ is the Jacobian determinant of the mapping evaluated in the solutions $x_i,y_i$ and is given by
\begin{align}\label{eqn_jacobian}
\left\lvert 
J(x,y) 
\right\rvert
=
4
\left\lvert 
\begin{bmatrix}
    b & -a
\end{bmatrix}\begin{bmatrix}
    x \\ y
\end{bmatrix}\right\rvert.
\end{align}
We can interpret \eqref{eqn_zw} as circles with radii $\sqrt{\mathbf{w}-1}$ and $\sqrt{\mathbf{z}}$ and centers in the origin and $\begin{bmatrix}
    a \\b
\end{bmatrix}$ respectively. Then $x_i, y_i$ are the points of intersection of the circles given by the following cases namely: 1) When $\vert r-\sqrt{w-1} \vert<\sqrt{z}$ and $z>0, w>1$, we have two intersection points which are given by \eqref{eqn_intersections}, 2) When $\vert r-\sqrt{w-1} \vert=\sqrt{z}$ and $z>0, w>1$, we get $c_{2}(z,w) = 0$ and it leads to an unique intersection at $\begin{bmatrix} x_1 & y_1 \end{bmatrix}^{\top} = c_{1}(z,w) \begin{bmatrix}a & b\end{bmatrix}^{\top}$. For all other cases there is no intersection. Inserting these results into \eqref{eqn_jacobian}, we obtain the following results for when two intersections exist.
\begin{align*}
&\left\lvert J(x_i(z,w),y_i(z,w))
\right\rvert
\\
&=
4
\left\lvert 
c_1(z,w)\begin{bmatrix}
    b & -a
\end{bmatrix}\begin{bmatrix}
    a \\ b
\end{bmatrix}\pm c_2(z,w)\begin{bmatrix}
    b & -a
\end{bmatrix}\begin{bmatrix}
    -b \\ a
\end{bmatrix}
\right\rvert \\
&=
4r^2c_2(z,w).  
\end{align*}
Inserting this into \eqref{eqn_cdf_transform} for each case yields \eqref{eqn_fzw}. \hfill $\square$

\subsection*{Proof of Lemma \ref{lemma_complex_ball_radius}}
\begin{proof}
We know that a circle which does not intersect the north pole on the Riemann sphere maps to another circle on the complex plane through the stereographic projection. Furthermore, a meridian on the Riemann sphere maps to a straight radial line. Take some circle on $\mathfrak R$ with chordal diameter $d\in[0,1]$ and some center point $\Bar{R}=(x,y,z)$. Now consider the meridian through $\Bar{R}$ and its intersection with the circle which yields two points $R_1$ and $R_2$ between which the chordal distance is $d$. Furthermore, $P_1\colon=\phi(R_1)$ and $P_2\colon=\phi(R_2)$ will lie on the opposite sides of the projected circle and thus at a distance $\hat{d}$ equal to the projected diameter. Consider $R_1=(\frac{1}{2}, \theta, \varphi)$ and $R_2=(\frac{1}{2}, \hat{\theta}, \varphi)$, with $\hat{\theta}=\theta+\Delta\theta$ to be the points described above given in spherical coordinates with its origin in $(x,y,z)=(0, 0, \frac{1}{2})$. The radial coordinate is fixed at $\frac{1}{2}$ due to the points being on $\mathfrak R$. They lie on the meridian at the fixed azimuth angle $\varphi$ and at a chordal distance $d$ from each other such that $\sin\frac{\Delta\theta}{2}=d.$ We know $\varphi$ only determines the direction of the projected diameter line but not its length. As such, $\hat{d}$ is invariant of $\varphi$ and we set it equal to zero. Now, equivalently consider the points in Cartesian coordinates to get
\begin{align*}
R_1=\frac{1}{2}\left(\sin\theta, 0, 1+\cos\theta\right), \quad
R_2=\frac{1}{2}\left(\sin\hat{\theta}, 0, 1+\cos\hat{\theta}\right).
\end{align*}
Projecting these onto $\mathbb{C}$ using the stereographic projection, we get real points 
\begin{align*}
    P_1=\frac{\sin\theta}{1-\cos\theta}, \quad   P_2=\frac{\sin\hat{\theta}}{1-\cos\hat{\theta}}.
\end{align*}
We then get 
\begin{equation*}
    \hat{d}=\left\lvert P_2-P_1\right\rvert=\left\lvert \frac{\sin{\left(\theta+\Delta\theta\right)}}{1-\cos{\left(\theta+\Delta\theta\right)}}- \frac{\sin\theta}{1-\cos\theta}\right\rvert
\end{equation*}
where 
\begin{align*}
\sin{\frac{\Delta\theta}{2}}=d,  \theta=\theta(\bar{R})
\text{ such that }\cos\left(\theta+\frac{\Delta\theta}{2}\right)=2z-1.
\end{align*}

Furthermore we know the center point is given by $\frac{P_1+P_2}{2}$. Since direction now plays a part we need to consider the azimuth angle $\varphi$ again, however we note that the distance from the origin to $P_1$ and $P_2$ respectively is still unaffected. We therefore find for any $\varphi\in[0,\;2\pi)$ that
\begin{align*}
    P_1&=\frac{\sin\theta}{1-\cos\theta}\; \left(\cos\varphi+\sin\varphi j\right), \\   
    P_2&=\frac{\sin\hat{\theta}}{1-\cos\hat{\theta}}\; \left(\cos\varphi+\sin\varphi j\right)
\end{align*}
where $\varphi=\varphi(\Bar{R})$ and the expression for the center point $\hat{P}$ follows from $\hat{P}=\frac{p_1+p_2}{2}$ yielding the result.
\end{proof}







\end{document}